\documentclass[a4paper,10pt,twoside]{article}
\usepackage[utf8]{inputenc}
\usepackage[english]{babel}
\usepackage[T1]{fontenc}
\usepackage{lmodern}
\usepackage{amsmath,amssymb,amsthm,mathtools} 

\newcommand{\im}{\mathrm{i}}

\newcommand{\Z}{\mathbb{Z}}
\newcommand{\R}{\mathbb{R}}
\newcommand{\C}{\mathbb{C}}
\newcommand{\defeq}{\coloneqq}
\newcommand{\lhs}{\langle\langle}
\newcommand{\rhs}{\rangle\rangle}
\newcommand{\cF}{\mathcal{F}}
\newcommand{\one}{\mathbf{1}}
\newcommand{\cA}{\mathcal{A}}
\newcommand{\cAe}{\mathcal{A}_{\mathrm{e}}}
\DeclareMathOperator{\tr}{tr}
\newcommand{\fa}{\dagger}
\newcommand{\la}{*}
\newcommand{\fge}{\mathfrak{g}_{\mathrm{e}}}
\newcommand{\fg}{\mathfrak{g}}
\newcommand{\ue}{\mathrm{U}}
\newcommand{\so}{\mathfrak{so}}
\newcommand{\fh}{\mathfrak{h}}
\newcommand{\fm}{\mathfrak{m}}
\newcommand{\fu}{\mathfrak{u}}
\newcommand{\fn}{\mathfrak{n}}
\newcommand{\cAn}{\mathcal{A}_0}
\newcommand{\Ge}{G_{\mathrm{e}}}

\newcommand{\vac}{\psi_0}
\newcommand{\acl}{[}
\newcommand{\acr}{]_{\mathrm{a.c.}}}
\newcommand{\wa}{\star}
\newcommand{\Wo}{\mathfrak{n}_\mathrm{s}}
\newcommand{\fp}{\mathfrak{p}}
\newcommand{\ad}{\mathrm{Ad}}
\newcommand{\xd}{\mathrm{d}}
\newcommand{\SO}{\mathrm{SO}}

\newcommand{\Galg}{G_{\mathrm{alg}}}
\newcommand{\Gdec}{G_{\mathrm{dec}}}
\newcommand{\tcl}{\mathcal{T}(L)}
\newcommand{\asl}{\Xi(L)}
\newcommand{\lai}{(}
\newcommand{\rai}{)_{\cA'}}
\newcommand{\cAx}{\mathcal{A}_{\mathrm{s}}}
\newcommand{\lop}{\|}
\newcommand{\rop}{\|_{\mathrm{op}}}
\newcommand{\hol}{\mathrm{Hol}}
\newcommand{\spec}{\mathrm{spec}}

\theoremstyle{definition}
\newtheorem{dfn}{Definition}[section]

\theoremstyle{plain}
\newtheorem{lem}[dfn]{Lemma}
\newtheorem{prop}[dfn]{Proposition}
\newtheorem{thm}[dfn]{Theorem}
\newtheorem{cor}[dfn]{Corollary}

\allowdisplaybreaks

\begin{document}

\begin{titlepage}
\title{\textbf{Coherent states\\ in the fermionic Fock space}}
\author{Robert Oeckl\footnote{email: robert@matmor.unam.mx}\\ \\
Institute for Quantum Gravity,\\
Friedrich-Alexander-Universität Erlangen-Nürnberg,\\
Staudtstraße 7, 91058 Erlangen, Germany\\ \\
Centro de Ciencias Matemáticas,\\
Universidad Nacional Autónoma de México,\\
C.P.~58190, Morelia, Michoacán, Mexico}
\date{UNAM-CCM-2014-1\\ 12 August 2014\\ 30 November 2014 (v2)}

\maketitle

\vspace{\stretch{1}}

\begin{abstract}
We construct the coherent states in the sense of Gilmore and Perelomov for the fermionic Fock space. Our treatment is from the outset adapted to the infinite-dimensional case. The fermionic Fock space becomes in this way a reproducing kernel Hilbert space of continuous holomorphic functions.

\end{abstract}

\vspace{\stretch{1}}
\end{titlepage}

\section{Introduction}

Coherent states have long played an important role in bosonic quantum field theory. While there is no agreed upon universal definition to the term ``coherent state'' these are generally states that (a) approximate classical behavior in a reasonable sense, (b) are parametrized by continuous sets, often related to a classical phase space, (c) minimize certain joint uncertainty relations, and (d) exhibit factorization of suitable correlation functions. Moreover, mathematically, they often (a) span a dense subspace of the state space, (b) satisfy completeness relations, (c) are eigenstates of the annihilation operator, (d) arise through the application of a shift operator on the ground state, and (e) yield reproducing properties.

In the simplest case (see e.g.\ \cite{ItZu:qft}), the state space is a Fock space $\cF$ over a Hilbert space $L$ and the coherent states are in correspondence to elements of $L$. Moreover, $\cF$ becomes a space of holomorphic functions on $L$ \cite{Bar:hilbanalytic}. When the space $L$ is identified with solutions of the classical equations of motions, the coherent state associated to an element of $L$ behaves in many respects similar to the corresponding classical solution.
Unfortunately, there is no fermionic analogue of these coherent states. When mimicking the construction of these states with fermions it is necessary to use Grassmann variables in their definition. The resulting objects are useful for calculations, especially manipulations with the path integral \cite{OhKa:cohfermip}. However, they are not states, i.e., not elements of the Hilbert space.

The approaches of Gilmore \cite{ZFG:coherent} and Perelomov \cite{Per:coherent} to the concept of coherent state are more fruitful for obtaining genuine states also in the fermionic case. In this approach a ``dynamical group'' $G$ is identified that acts unitarily on the state space. The coherent states are then all the states generated by the action of the group on a reference state, usually the vacuum. The coherent states can be brought into correspondence with the elements of the homogeneous space $G/H$, where $H$ is the subgroup of $G$ that maps the reference state to a multiple of itself.
We shall consider here the case of a fermionic Fock space $\cF$ over the Hilbert space $L$, with the usual CAR-algebra of creation and annihilation operators. Restricting to the states of even degree and supposing that $L$ is finite-dimensional the relevant dynamical group $G$ and its complexification was analyzed in some detail by Balian and Brezin \cite{BaBr:nbogtrans}. If the dimension of $L$ is $n$ this group can be identified with $\SO(2n)$. The corresponding coherent states were introduced by Berezin who also analyzed the structure of the relevant homogeneous space $G/H$ where $H$ can be identified with $\mathrm{U}(n)$. In particular, Berezin shows that $G/H$ is a Kähler manifold. Even though Berezin's treatment is also based on matrices, he does give consideration to the case where $\cF$ is infinite-dimensional.
There is very little literature on coherent states for the full Fock space (including states of odd degree). The corresponding dynamical group is $\SO(2n+1)$ and the homogeneous space is $\SO(2n+1)/\mathrm{U}(n)$ if $L$ has dimension $n$ \cite{ZFG:coherent}.

The purpose of the present paper is to provide a treatment of fermionic coherent states that makes them fully accessible to quantum field theory, i.e., is from the outset adapted to the infinite-dimensional case. Moreover, the focus is on coherent states for the full Fock space, not only its even part. This is in view of making them available for applications such as outlined in \cite{Ber:grossneveu} or in \cite{ZFG:coherent}, e.g., Gross-Neveu models and the Hartree-Fock-Bogoliubov method, but in an infinite-dimensional setting. However, the paper is strictly limited to the development of the mathematical theory, with such applications being beyond its scope.

In Section~\ref{sec:CAR} the dynamical Lie algebra is defined based on the CAR-algebra. Conjugate linear maps play an important role. Exponentiation of this Lie algebra is to provide the dynamical group. In Section~\ref{sec:dyneven} the even dynamical Lie algebra is further analyzed and represented on the creation and annihilation operators. This Section is very much in the spirit of reference \cite{BaBr:nbogtrans}. In Section~\ref{sec:dynall} the analysis is extended to the full dynamical Lie algebra. As a crucial ingredient categorical bosonization \cite{Maj:cpbos} is added. In Section~\ref{sec:gparam} the dynamical group is defined and a decomposition theorem proven which generalizes the one of reference \cite{BaBr:nbogtrans}. The homogeneous space $G/H$ and a complexified version $G^\C/X_+$ are analyzed and parametrized in Section~\ref{sec:homspace}. Finally, coherent states are constructed in Section~\ref{sec:cohspace}. The Section culminates in a presentation of the Fock space $\cF$ as a reproducing kernel Hilbert space of continuous (anti-)holomorphic functions on another Hilbert space.

\section{CAR-algebra and dynamical Lie algebras}
\label{sec:CAR}

Let $L$ be a complex separable Hilbert space. Denote its inner product by $\{\cdot,\cdot\}$.\footnote{We use physicist's conventions throughout according to which a complex inner product is complex linear in the second argument and complex conjugate linear in the first argument.} For a bounded operator $X:L\to L$ we write the adjoint as $X^\la$. Denote by $\tcl$ the algebra of \emph{trace class operators} on $L$. For $\lambda\in \tcl$ denote its trace by,
\begin{equation}
 \tr(\lambda)\defeq\sum_{i\in I} \{\zeta_i, \lambda\zeta_i\},
\end{equation}
where $\{\zeta_i\}_{i\in I}$ is an orthonormal basis of $L$.

We consider the unital $*$-algebra $\cA$ over $\C$ generated by elements $a_{\xi}$ for $\xi\in L$ and subject to the relations,
\begin{equation}
 a_{\xi+\tau}=a_\xi+a_\tau,\quad a_{\lambda\xi}=\lambda a_{\xi},\quad
 a_{\xi} a_{\tau} + a_{\tau} a_{\xi}=0,\quad
 a_{\xi}^\fa a_{\tau}+a_{\tau} a_{\xi}^\fa=\{\xi,\tau\}\one .
\label{eq:CAR}
\end{equation}
This is the usual CAR-algebra over $L$, with $\fa$ denoting the $*$-structure. The elements $a_{\xi}$ are called the \emph{annihilation operators} and the elements $a_{\xi}^\fa$ are called the \emph{creation operators}. $\cA$ is $\Z$-graded by considering a monomial consisting of $p$ creation operators and $q$ annihilation operators to have degree $p-q$. Of particular interest are the $*$-subalgebras $\cAe$, consisting of elements of even degree, and $\cAn$, consisting of elements of degree zero.
Denote by $\cF$ the Fock space over $L$ on which $\cA$ acts in the usual way. Let $\cA'$, $\cAe'$, $\cAn'$ denote the completions of $\cA$, $\cAe$, $\cAn$ in the operator norm topology.

Given $\lambda\in\tcl$ define the operator $\hat{\lambda}:\cF\to\cF$ by,
\begin{equation}
  \hat{\lambda}\defeq \frac{1}{2} \sum_{i\in I}
   \left(a_{\zeta_i}^{\fa} a_{\lambda(\zeta_i)} - a_{\lambda(\zeta_i)} a_{\zeta_i}^{\fa}\right)
   = \sum_{i\in I} a_{\zeta_i}^{\fa} a_{\lambda(\zeta_i)}
     - \frac{1}{2}\tr(\lambda)\one .
\end{equation}
Up to a constant multiple of the identity this is the fermionic current operator as given in reference \cite{HaKa:aqft}. Note, the relation of the $*$-structures, $\hat{\lambda^\la}=\hat{\lambda}^\fa$. The vector space $\fh^\C$ of these current operators forms a complex Lie algebra, with the Lie bracket given by the commutator. This is,
\begin{align}
 & [\hat{\lambda},\hat{\lambda'}]=\hat{\lambda''}, && \text{with}\quad \lambda''=\lambda'\lambda-\lambda\lambda' .
 \label{eq:hcrel}
\end{align}
Thus $\tcl$ is naturally anti-isomorphic to $\fh^\C$ as a complex Lie algebra.
$\fh^\C$ has a natural real structure coming from its action on $\cF$. Namely, the complex conjugate of $X\in\fh^\C$ is given by $-X^\fa$. The real Lie subalgebra $\fh$ is thus generated by the elements $\hat{\lambda}$ that are skew-adjoint, equivalent to $\lambda$ itself being skew-adjoint, $\lambda=-\lambda^\la$. If $L$ is of dimension $n$, $\fh$ is the unitary Lie algebra $\fu(n;\R)$ and $\fh^\C$ is its complexification $\fu(n;\C)$. If $L$ is infinite-dimensional, $\fh$ is an infinite-dimensional unitary Lie algebra, further determined here by the condition that the operators $\lambda$ be trace class.
\begin{prop}
The completed universal enveloping algebra $\ue'(\fh)^\C=\ue'(\fh^\C)$ coincides with $\cAn'$.
\end{prop}
We shall refer to $\fh$ as the \emph{degree-preserving dynamical Lie algebra}.

Let $\Lambda:L\to L$ be a real linear map. We shall say that $\Lambda$ is \emph{anti-symmetric} (not to be confused with skew-adjoint) iff for all $\xi,\tau\in L$ we have,
\begin{equation}
 \{\xi,\Lambda\tau\}=-\{\tau,\Lambda\xi\} .
\end{equation}
This implies in particular that $\Lambda$ is conjugate linear rather than complex linear. Moreover, $\Lambda^2$ is complex linear, self-adjoint, and strictly negative (if $\Lambda\neq 0$),
\begin{gather}
 \{\xi,\Lambda^2\tau\}=-\{\Lambda\tau,\Lambda\xi\}=\{\Lambda^2\xi,\tau\},\\
 \{\xi,\Lambda^2\xi\}=-\{\Lambda\xi,\Lambda\xi\}\le 0 .
\end{gather}
Let $\asl$ be the vector space of anti-symmetric maps $\Lambda:L\to L$ such that $\Lambda^2$ is trace class (this implies that $\Lambda$ is bounded).\footnote{Rather than requiring that $\Lambda^2$ be trace class we could impose an equivalent condition on $\Lambda$ itself. This would be a real linear version of the Hilbert-Schmidt property.}
Even though the elements of $\asl$ are not complex linear maps, the space $\asl$ itself is naturally a complex vector space. This is because given $\Lambda\in\asl$, the map $\im\Lambda$ defined by $(\im \Lambda)(\eta)\defeq \im\, \Lambda(\eta)$ is also in $\asl$. But note, $\im\, \Lambda(\eta)=-\Lambda(\im\eta)$.
Given $\Lambda\in\asl$ define the operator $\hat{\Lambda}:\cF\to\cF$ as follows,
\begin{equation}
  \hat{\Lambda}\defeq \frac{1}{2} \sum_{i\in I}
   a_{\zeta_i} a_{\Lambda(\zeta_i)} .
\end{equation}
The condition on $\Lambda$ to be anti-symmetric is precisely equivalent to this definition being manifestly basis independent. These operators form a complex abelian Lie algebra that we shall denote by $\fm_+$. The adjoint of $\hat{\Lambda}$ is,
\begin{equation}
  \hat{\Lambda}^\fa= \frac{1}{2} \sum_{i\in I}
   a_{\Lambda(\zeta_i)}^\fa a_{\zeta_i}^\fa  .
\end{equation}
The complex abelian Lie algebra spanned by these operators will be denoted by $\fm_-$. The space $\fm_+$ is naturally isomorphic to $\asl$ as a complex vector space and $\fm_-$ is naturally isomorphic to $\overline{\asl}$ as a complex vector space. Here $\overline{\asl}$ denotes $\asl$ but with opposite complex structure.
For later use we also note the operator norm,
\begin{equation}
 \lop \hat{\Lambda}\rop^2=\lop \hat{\Lambda}^\fa\rop^2=\|\hat{\Lambda}^\fa\vac\|_{\cF}^2=-\frac{1}{2}\tr(\Lambda^2)
\label{eq:Lopn}
\end{equation}
The second equality arises here from the fact that $\hat{\Lambda}^\fa$ does not contain any annihilation operators.
Denote the direct sum by $\fm^\C\defeq\fm_+\oplus\fm_-$. Its real subspace spanned by the skew-adjoint operators $\hat{\Lambda}-\hat{\Lambda}^\fa$ denote by $\fm$. Then the direct sum $\fge\defeq \fh\oplus\fm$ is a real Lie algebra with complexification $\fge^\C=\fh^\C\oplus \fm^\C$. In addition to (\ref{eq:hcrel}) the Lie brackets are as follows. (The brackets not mentioned vanish.)
\begin{align}
 & [\hat{\lambda},\hat{\Lambda}]=\hat{\Lambda'},
  && \text{with}\quad \Lambda'=-\lambda\Lambda-\Lambda\lambda^\la, \label{eq:mrel1} \\
 & [\hat{\lambda},\hat{\Lambda}^\fa]=\hat{\Lambda'}^\fa,
  && \text{with}\quad \Lambda'= \lambda^\la\Lambda+\Lambda\lambda, \label{eq:mrel2} \\
 & [\hat{\Lambda'},\hat{\Lambda}^\fa]=\hat{\lambda},
  && \text{with}\quad \lambda=\Lambda'\Lambda . \label{eq:mrel3}
\end{align}

\begin{prop}
The completed universal enveloping algebra $\ue'(\fge)^\C=\ue'(\fge^\C)$ coincides with the even algebra $\cAe'$.
\end{prop}
We shall refer to $\fge$ as the \emph{even dynamical Lie algebra}.

Set $\hat{\xi}\defeq 1/\sqrt{2}\, a_{\xi}$ and consider the complex vector spaces spanned by the operators $\hat{\xi},\hat{\xi}^\fa$ denoting them by $\fn_+,\fn_-$ respectively. With this definition
\begin{equation}
\lop\hat{\xi}\rop^2=\lop\hat{\xi}^\fa\rop^2=\|\hat{\xi}^\fa\vac\|_{\cF}^2=\frac{1}{2}\|\xi\|_L^2 .
\label{eq:xopn}
\end{equation}
$\fn_+$ is naturally isomorphic to $L$ as a complex vector space while $\fn_-$ is naturally isomorphic to $\overline{L}$ as a complex vector space. Here $\overline{L}$ denotes $L$ with opposite complex structure. Denote the direct sum by $\fn^\C\defeq\fn_+\oplus\fn_-$ and its real subspace spanned by $\hat{\xi}-\hat{\xi}^\fa$ by $\fn$. Then, $\fg\defeq \fh\oplus\fm\oplus\fn$ forms a real Lie algebra and $\fg^\C= \fh^\C\oplus\fm^\C\oplus\fn^\C$ is its complexification. In addition to (\ref{eq:hcrel}), (\ref{eq:mrel1}), (\ref{eq:mrel2}), (\ref{eq:mrel3}) the Lie brackets are as follows. (The brackets not mentioned vanish.)
\begin{align}
 & [\hat{\lambda},\hat{\xi}]=\hat{\xi'},
  && \text{with}\quad \xi'=-\lambda\xi, \label{eq:nrel1} \\
 & [\hat{\lambda},\hat{\xi}^\fa]=\hat{\xi'}^\fa,
  && \text{with}\quad \xi'= \lambda^\la\xi, \label{eq:nrel2} \\
 & [\hat{\Lambda}^\fa,\hat{\xi}]=\hat{\xi'}^\fa,
  && \text{with}\quad \xi'=\Lambda\xi , \label{eq:nrel3} \\
 & [\hat{\Lambda},\hat{\xi}^\fa]=\hat{\xi'},
  && \text{with}\quad \xi'=-\Lambda\xi , \label{eq:nrel4} \\
 & [\hat{\xi},\hat{\xi'}]=\hat{\Lambda},
  && \text{with}\quad \Lambda \eta=\xi'\{\eta,\xi\}-\xi\{\eta,\xi'\},
  \label{eq:nrel5} \\
 & [\hat{\xi}^\fa,\hat{\xi'}^\fa]=\hat{\Lambda}^\fa,
  && \text{with}\quad \Lambda \eta=\xi\{\eta,\xi'\}-\xi'\{\eta,\xi\},
  \label{eq:nrel6} \\
 & [\hat{\xi}^\fa,\hat{\xi'}]=\hat{\lambda},
  && \text{with}\quad \lambda \eta=\xi'\{\xi,\eta\}.
  \label{eq:nrel7}
\end{align}

\begin{prop}
The completed universal enveloping algebra $\ue'(\fg)^\C=\ue'(\fg^\C)$ coincides with the CAR-algebra $\cA$.
\end{prop}
We shall refer to $\fg$ as the \emph{(full) dynamical Lie algebra}.

\section{Adjoint representation\\ of even dynamical Lie algebra}
\label{sec:dyneven}

The Lie algebra $\fg^\C$ acts on $\cA'$ via the \emph{adjoint action}. For $x\in\fg^\C$ we denote its adjoint action on $a\in\cA'$ by
\begin{equation}
\check{x}a\defeq [x,a] .
\end{equation}
Let $\langle\cdot,\cdot\rangle$ denote the inner product in the Fock space $\cF$ and $\vac$ denote its vacuum state. Define an inner product $\lai\cdot,\cdot\rai$ on $\cA$ via,
\begin{equation}
 \lai x,y\rai\defeq \frac{1}{2}\langle\vac, (x^\fa y +y x^\fa)\vac\rangle
\label{eq:defipA}
\end{equation}
Denote the adjoint for operators on $\cA$ with respect to this inner product by $^\wa$. We note that under this inner product the subspaces $\fn_-\subset\cA$ and $\fn_+\subset\cA$ are orthogonal. On $\fn^\C=\fn_+\oplus\fn_-$ the inner product takes the form,
\begin{equation}
 \lai\hat{\xi}^\fa+\hat{\eta},\hat{\mu}^\fa+\hat{\tau}\rai=\{\mu,\xi\}+\{\eta,\tau\} .
\end{equation}
In particular, it is positive definite.

We proceed to examine the action of the complexified degree-preserving dynamical Lie algebra $\fh^\C$.
\begin{prop}
Under the adjoint action of $\fh^\C$ on $\cA$ the subspaces $\fn_-\subset\cA$ and $\fn_+\subset\cA$ are invariant. Both are faithful and irreducible representations. Moreover they are $*$-representations with respect to the inner product $\lai\cdot,\cdot\rai$. More explicitly, $\fh^\C$ acts on each of $\fn_-\subset\cA$ and $\fn_+\subset\cA$ precisely as the Lie algebra of all trace class operators. $\fh$ acts as the Lie algebra of all skew-adjoint trace class operators.
\end{prop}
\begin{proof}
This is straightforward to verify from relations (\ref{eq:nrel1}) and (\ref{eq:nrel2}) and definition (\ref{eq:defipA}).
\end{proof}
The $*$-representation property means that for $x\in \fh^\C$ we have $\check{x^\fa}=\check{x}^\wa$.
The inner product on $\fn_+$ (say) induces a Hilbert-Schmidt inner product on $\fh^\C$. Given $x,y\in \fh^\C$ define,
\begin{equation}
 \lhs x,y\rhs \defeq 2\tr_{\fn_+}(\check{x}^\wa \check{y}) .
\label{eq:hsiph}
\end{equation}
(The factor of $2$ is a useful convention here.) Replacing $\fn_+$ by $\fn_-$ yields the same inner product. For $\lambda\in\tcl$ this is,
\begin{equation}
 \lhs \check{\lambda_1},
  \check{\lambda_2} \rhs
 =2\tr_L(\lambda_1^\la \lambda_2) .
\label{eq:iph}
\end{equation}
Note that we write $\check{\lambda}$ rather than $\check{\hat{\lambda}}$ for simplicity.
This shows in particular that this inner product is positive definite. Moreover, it is easy to see that it is invariant under the action of an element $z$ of the real Lie algebra $\fh$, i.e., with $\check{z}^\wa=-\check{z}$,
\begin{multline}
 \lhs \check{z} x,y\rhs+\lhs x, \check{z} y\rhs
 =\lhs [z,x],y\rhs + \lhs x, [z,y] \rhs\\
 =2\tr_{\fn_+}(\check{x}^\wa\check{z}^\wa\check{y}-\check{z}^\wa\check{x}^\wa\check{y})+2\tr_{\fn_+}(\check{x}^\wa\check{z}\check{y}-\check{x}^\wa\check{y}\check{z})
 =2\tr_{\fn_+}(\check{z}\check{x}^\wa\check{y}-\check{x}^\wa\check{y}\check{z})=0 .
\label{eq:invrla}
\end{multline}
If $L$ is finite-dimensional
$\fh$ is a compact real Lie algebra and thus admits a unique negative definite invariant bilinear form, called the \emph{Killing form}. Thus, the inner product $\lhs\cdot,\cdot\rhs$ restricted to $\fh$ is precisely a negative multiple of this Killing form.

We turn to examine the action of the complexified even dynamical Lie algebra $\fge^\C$.
\begin{prop}
Under the adjoint action of $\fge^\C$ the subspace $\fn^\C\subset\cA'$ is invariant and forms a faithful and irreducible representation. Moreover it is a $*$-representation with respect to the inner product $\lai\cdot,\cdot\rai$. The representation of $\fge$ on $\fn^\C$ is faithful, irreducible and unitary.
\end{prop}
\begin{proof}
This follows with relations (\ref{eq:nrel1}), (\ref{eq:nrel2}), (\ref{eq:nrel3}), (\ref{eq:nrel4}) and definition (\ref{eq:defipA}).
\end{proof}
Using the convenient notation $(\xi,\eta)\defeq \hat{\xi}^\fa+\hat{\eta}$ for elements of $\fn^\C$, the action of $\fge^\C$ takes the explicit form
\begin{equation}
(\check{\lambda}+\check{\Lambda}^\fa+\check{\Lambda'}) (\tau,\eta)=
 (\lambda^*\tau+\Lambda\eta,-\lambda\eta-\Lambda\tau) .
 \label{eq:fgeact}
\end{equation}

We proceed to examine in more detail $\fn^\C=\fn_+\oplus\fn_-$ as a subspace of $\cA'$. The inner product on $\fn^\C$ also allows to express the invariance of the CAR (\ref{eq:CAR}) under the action of a group or Lie algebra. For the action of a Lie algebra element $Z$ this invariance is,
\begin{equation}
 \acl \check{Z} x, y\acr + \acl x, \check{Z} y\acr=0,\qquad\forall x,y\in \fn^\C .
\end{equation}
Here $\acl\cdot,\cdot\acr$ denotes the anti-commutator. From the definition (\ref{eq:defipA}) this is immediately seen to be equivalent to,
\begin{equation}
 \lai (\check{Z} x)^\fa,y\rai+\lai x^\fa,\check{Z} y\rai=0,\qquad\forall x,y\in \fn^\C .
\end{equation}
This in turn can be rewritten as,
\begin{equation}
 Z^\wa=-\sigma Z \sigma,
\label{eq:lainvact}
\end{equation}
where the conjugate linear map $\sigma:\fn^\C\to\fn^\C$ is a different notation for the operation of taking the adjoint in $\cA$. The validity of (\ref{eq:lainvact}) for elements of $\fge^\C$ follows from the $*$-representation property, i.e., $\check{x^\fa}=\check{x}^\wa$, together with the relation between commutator and $*$-structure
\begin{equation}
 [x^\fa,a]=-[x,a^\fa]^\fa\quad\forall a\in \fn^\C .
\end{equation}
If $L$ is finite-dimensional any operator $Z$ on $\fn^\C$ satisfying equation (\ref{eq:lainvact}) is easily seen to be of the form (\ref{eq:fgeact}) with $\lambda$ complex linear and $\Lambda,\Lambda'$ anti-symmetric. Thus, equation (\ref{eq:lainvact}) completely determines the Lie algebra $\fge^\C$. In particular, $\fge^\C$ is precisely the Lie algebra of transformations that preserve the anti-commutation relations. This prompted Balian and Brezin to call the associated Lie group $\Ge^\C$ the group of ``generalized Bogoliubov transformations'' \cite{BaBr:nbogtrans}. The usual Bogoliubov transformations arise by restriction to the real Lie algebra $\fge$ and associated real Lie group $\Ge$ that acts unitarily, both on $\cF$ and on $\fn^\C$. We summarize:
\begin{prop}
The Lie algebra $\fge^\C$ acts on $\fn^\C$ through operators satisfying equation (\ref{eq:lainvact}). If $L$ is finite-dimensional any operator on $\fn^\C$ satisfying equation (\ref{eq:lainvact}) arises in this way.
\end{prop}

In the specific case that $L$ has finite dimension $n$ we may identify operators on $\fn^\C$ with complex $2n\times 2n$ matrices. Taking care of the conjugate linearity of $\sigma$, the equation (\ref{eq:lainvact}) for operators translates for matrices to the equation,
\begin{equation}
 Z^\wa=-\tilde{\sigma}\overline{Z}\tilde{\sigma} .
 \label{eq:fgesu}
\end{equation}
Here $\overline{Z}$ is the entry-wise complex conjugation of the matrix $Z$. $\tilde{\sigma}$ is, in block matrix form,
\begin{equation}
\tilde{\sigma}=\begin{pmatrix} 0 & \one_{n\times n} \\ \one_{n\times n} & 0
\end{pmatrix} .
\end{equation}
The equation (\ref{eq:fgesu}) is easily recognizable as determining the special orthogonal matrix Lie algebra. This recovers the well known fact that $\fge^\C$ is precisely $\so(2n;\C)$ (and $\fge$ is $\so(2n;\R)$) \cite{ZFG:coherent}.

The inner product on $\fn^\C$ induces a Hilbert-Schmidt inner product on $\fge^\C$. Given $x,y\in \fge^\C$ define,
\begin{equation}
 \lhs x,y\rhs \defeq\tr_{\fn^\C}(\check{x}^\wa \check{y}) .
\end{equation}
For generic elements of $\fge^\C$, parametrized as above, this yields,
\begin{equation}
 \lhs \check{\lambda_1}+\check{\Lambda_1}^\wa+\check{\Lambda_1'},
  \check{\lambda_2}+\check{\Lambda_2}^\wa+\check{\Lambda_2'} \rhs
 =2\tr_L(\lambda_1^\la \lambda_2)-\tr_L(\Lambda_2'\Lambda_1')-\tr_L(\Lambda_1\Lambda_2) .
\label{eq:ipe}
\end{equation}
Note that we have normalized the corresponding inner product on $\fh^\C$ given by (\ref{eq:hsiph}) so as to precisely coincide with this one when restricted. The inner product is positive definite. Moreover, it is easy to see that it is invariant under the action of an element $z$ of the real Lie algebra $\fge$. This follows analogous to (\ref{eq:invrla}). If $L$ is finite-dimensional, $\fge$ is a compact real Lie algebra and admits a unique negative definite invariant bilinear form. Thus, the inner product $\lhs\cdot,\cdot\rhs$ restricted to $\fge$ is precisely a negative multiple of this Killing form.

\section{The full dynamical Lie algebra\\ and bosonization}
\label{sec:dynall}

If $L$ is finite-dimensional of dimension $n$ the full dynamical Lie algebra $\fg$ is known to be $\so(2n+1;\R)$ and an explicit representation in terms of $(2n+1)\times (2n+1)$ matrices can be written down \cite{ZFG:coherent}. However, this representation does not arise in the way that the representation on $\fn^\C\subset\cA'$ was obtained for $\fge$. Indeed, acting with $\fg$ does not leave $\fn^\C$ invariant, but rather leads to a much larger representation as can be seen from relations (\ref{eq:nrel5}), (\ref{eq:nrel6}) and (\ref{eq:nrel7}). While the analog of this approach works for the full dynamical group in the bosonic case (see \cite{BaBr:nbogtrans}), we should not be surprised that it fails here, once we lift the restriction to the even Lie subalgebra $\fge\subset\fg$. After all, the CAR-algebra $\cA$ should really be considered a superalgebra and we should use super-Lie algebras and supergroups to properly deal with this. On the other hand, this would defeat our declared purpose of constructing ordinary spaces (i.e., sets) of coherent states rather than superspaces.

There is a way out, however, and that is \emph{categorical bosonization} \cite{Maj:cpbos}. In general this is a functor from a category of braided objects to a category of unbraided objects. In the present case the braiding is simply a $\Z_2$-grading. Bosonization then converts $\Z_2$-graded objects, i.e., superobjects to ungraded objects. In the case at hand the superalgebra $\cA$ is replaced by the ordinary algebra $\cAx$. $\cAx$ is the same algebra as $\cA$, but with an additional element $k$ adjoint, satisfying the relations,\footnote{$\cAx$ really arises as a smash product $\cA \rtimes\Z_2$, where $\Z_2$ is the two-dimensional Hopf algebra spanned by the unit and the $\Z_2$-grading operator.}
\begin{equation}
 k^2=\one,\quad k^\fa=k,\quad k a_\xi=- a_\xi k,\quad a_\xi^\fa k =-k a_\xi^\fa .
\label{eq:krel}
\end{equation}
On $\cF$ the element $k$ acts as the identity on states of even degree and as minus the identity on states of odd degree.

Consider now the subspace $k \fn^\C$ of $\cA$. It is easy to see that $\fn^\C\to k \fn^\C$ given by $x\mapsto k x$ is an isomorphism both of Hilbert spaces with the inner product (\ref{eq:defipA}) and of representations of $\fge$. Note, however, that the map $\sigma': k\fn^\C\to k\fn^\C$ induced by this isomorphism is not the adjoint, but minus the adjoint. But this again leads to formula (\ref{eq:lainvact}), with $\sigma$ replaced by $\sigma'$. Indeed, using for elements of $k\fn^\C$ the notation $(\xi,\eta)\defeq k\hat{\xi}^\fa+k\hat{\eta}$ yields $\sigma':(\xi,\eta)\mapsto (\eta,\xi)$. Moreover, the action of $\fge$ is expressed by the very same formula (\ref{eq:fgeact}) as in the case of $\fn^\C$. To be explicit, we list the relevant commutation relations,
\begin{align}
 & [\hat{\lambda},k\hat{\xi}]=k\hat{\xi'},
  && \text{with}\quad \xi'=-\lambda\xi, \label{eq:nrelo1} \\
 & [\hat{\lambda},k\hat{\xi}^\fa]=k\hat{\xi'}^\fa,
  && \text{with}\quad \xi'= \lambda^\la\xi, \label{eq:nrelo2} \\
 & [\hat{\Lambda}^\fa,k\hat{\xi}]=k\hat{\xi'}^\fa,
  && \text{with}\quad \xi'=\Lambda\xi , \label{eq:nrelo3} \\
 & [\hat{\Lambda},k\hat{\xi}^\fa]=k\hat{\xi'},
  && \text{with}\quad \xi'=-\Lambda\xi . \label{eq:nrelo4}
\end{align}

The difference to the representation $\fn^\C\subset\cA'$ becomes obvious once we act with the additional elements of $\fn^\C\subset\fg^\C$. The representation grows by only a $1$-dimensional space, generated by $k$.
\begin{prop}
Under the adjoint action of $\fg^\C$ on $\cAx'$ the subspace $\Wo\defeq k \fn^\C\oplus\C[k]\subset\cAx'$ is invariant and forms a faithful and irreducible representation. Moreover it is a $*$-representation with respect to the inner product $\lai\cdot,\cdot\rai$. The representation of $\fg$ on $\Wo$ is faithful, irreducible and unitary.
\end{prop}
\begin{proof}
This follows with relations (\ref{eq:krel}), (\ref{eq:nrelo1}), (\ref{eq:nrelo2}), (\ref{eq:nrelo3}), (\ref{eq:nrelo4}) and definition (\ref{eq:defipA}).
\end{proof}
Explicitly, the additional non-vanishing relations are,
\begin{align}
 & [\hat{\xi},k\hat{\xi'}^\fa]=-k \acl \hat{\xi},\hat{\xi'}^\fa\acr
  =-\frac{1}{2}\{\xi',\xi\} k,
  \label{eq:nrelo5} \\
 & [\hat{\xi}^\fa,k\hat{\xi'}]=-k \acl \hat{\xi}^\fa,\hat{\xi'}\acr
  =-\frac{1}{2}\{\xi,\xi'\} k,
  \label{eq:nrelo6} \\
 & [\hat{\xi},k]=-2 k \hat{\xi},
  \label{eq:nrelo7} \\
 & [\hat{\xi}^\fa,k]=-2 k \hat{\xi}^\fa.
  \label{eq:nrelo8}
\end{align}
We shall use the notation $(\mu,\tau,\eta)\defeq 1/2 \mu k + k \hat{\tau}^\fa+ k\hat{\eta}$ for elements of $\Wo$. Let $\lambda\in\tcl$, $\Lambda,\Lambda'\in\asl$, and $\xi,\xi'\in L$. The action of a general element of $\fg^\C$ on $\Wo$ then takes the form,
\begin{multline}
(\check{\lambda}+\check{\Lambda}^\fa+\check{\Lambda'}+\check{\xi}^\fa+\check{\xi'})(\mu,\tau,\eta)=\\
 (-\{\xi,\eta\}-\{\tau,\xi'\},\lambda^*\tau+\Lambda\eta-\overline{\mu}\tau,-\lambda\eta-\Lambda\tau-\mu\eta) .
 \label{eq:fgact}
\end{multline}
The inner product (\ref{eq:defipA}) is positive-definite also on $\Wo$. Moreover, as for $\fn^\C$ the equality of the adjoint with respect to the inner products in $\cF$ and in $\Wo$ is easily verified. Let $\sigma:\Wo\to\Wo$ be again the operation of taking the adjoint. Then, equation (\ref{eq:lainvact}) again characterizes precisely the Lie algebra $\fg^\C$ in the following sense. The action of any element $Z$ of $\fg^\C$ satisfies the equation and conversely, if $L$ is finite-dimensional, any operator $Z$ on $\Wo$ satisfying equation (\ref{eq:lainvact}) must be of the form $Z=\check{\lambda}+\check{\Lambda}^\fa+\check{\Lambda'}+\check{\xi}^\fa+\check{\xi'}$ with $\lambda,\Lambda,\Lambda',\xi,\xi'$ as above. Summarizing:
\begin{prop}
The Lie algebra $\fg^\C$ acts on $\Wo$ through operators satisfying equation (\ref{eq:lainvact}). If $L$ is finite-dimensional any operator on $\fn^\C$ satisfying equation (\ref{eq:lainvact}) arises in this way.
\end{prop}

If $L$ is finite-dimensional, of dimension $n$, we may identify operators on $\Wo$ with $(2n+1)\times(2n+1)$ matrices. Equation (\ref{eq:lainvact}) then turns into equation (\ref{eq:fgesu}) with the matrix $\tilde{\sigma}$ in a suitable basis given by
\begin{equation}
\tilde{\sigma}=\begin{pmatrix} 1 & 0 & 0 \\ 0 & 0 & \one_{n\times n} \\ 0 & \one_{n\times n} & 0
\end{pmatrix} .
\end{equation}
This makes then explicit the structure of $\fg^\C$ as $\so(2n+1;\C)$.

As for the even dynamical Lie algebra $\fge^\C$, the inner product on $\Wo$ induces a Hilbert-Schmidt inner product on $\fg^\C$. Given $x,y\in \fg^\C$ define,
\begin{equation}
 \lhs x,y\rhs \defeq\tr_{\Wo}(\check{x}^\wa \check{y}) .
\label{eq:gip}
\end{equation}
For generic elements of $\fg^\C$, parametrized as above, this yields,
\begin{multline}
 \lhs \check{\lambda_1}+\check{\Lambda_1}^\wa+\check{\Lambda_1'}+\check{\xi_1}^\wa+\check{\xi_1'},
  \check{\lambda_2}+\check{\Lambda_2}^\wa+\check{\Lambda_2'}+\check{\xi_2}^\wa+\check{\xi_2'} \rhs \\
 =2\tr_L(\lambda_1^\la \lambda_2)-\tr_L(\Lambda_2'\Lambda_1')-\tr_L(\Lambda_1\Lambda_2) +2\{\xi_2,\xi_1\} +2\{\xi_1',\xi_2'\}.
\label{eq:gipeval}
\end{multline}
This makes positive definiteness manifest and also shows that restriction to $\fge^\C$ recovers the inner product (\ref{eq:ipe}). Invariance under the action of the real Lie algebra $\fg$ follows analogous to equation (\ref{eq:invrla}). Thus, the inner product $\lhs\cdot,\cdot\rhs$ restricted to $\fg$ is a negative multiple of its Killing form.

\section{Parametrization of the dynamical group}
\label{sec:gparam}

The groups $H$, $\Ge$, $G$ corresponding to the Lie algebras $\fh$, $\fge$, $\fg$ shall be called the \emph{degree-preserving}, \emph{even} or \emph{(full) dynamical group} respectively. We shall also be interested in their complexified versions corresponding to the complexified Lie algebras.

We define $H^\C$ to be the set of operators on $\cF$ arising from the exponentiation of $\hat{\lambda}\in\fh^\C$ where $\lambda\in\tcl$.
\begin{prop}
Under the adjoint representation $H^\C$ acts precisely as the group of all invertible operators on $\fn_+$ of the form $\one_{\fn_+}+x$ where $x:\fn_+\to \fn_+$ is of trace class.
\end{prop}
\begin{proof}
The adjoint action of $\exp(\hat{\lambda})$ on $\fn_+$ is given by $\exp(\check{\lambda}) \hat{\xi}=\hat{\xi'}$ with $\xi'=\exp(-\lambda)\xi$. In the following we identify $L$ with $\fn_+$ via $\xi\mapsto\hat{\xi}$ for ease of notation. Set $x\defeq \exp(-\lambda)-\one_L$. By the holomorphic functional calculus \cite{Sch:pfunctanal} $\spec(x)=\spec(\exp(-\lambda)-\one_L)=\exp(-\spec(\lambda))-1$, where $\spec$ denotes the spectrum. $\lambda$ being compact, its spectrum is bounded and discrete with only the origin a possible accumulation point. Thus $x$ also has a bounded and discrete spectrum with only the origin as accumulation point. Moreover, $\lambda$ being trace class the sum over its spectrum (with multiplicities) converges absolutely. So does then the sum over the spectrum of $x$ as the two spectra approach each other near the origin. Thus, $x$ is also trace class. Recall also that by construction $\one_L+x$ is invertible. Conversely, suppose now $x\in\tcl$ such that $\one_L+x$ is invertible. In particular, there is some neighborhood of the origin which does not contain any point of the spectrum of $\one_L+x$. Also the only possible accumulation point of $\spec(\one_L+x)$ is away from the origin. So we can take a logarithm of $\one_L+x$. What is more, any branch of the logarithm that coincides with the standard branch near $1$ will map a neighborhood of $1$ to a neighborhood of $0$. By the same argument as before, $\log(\one_L+x)$ is then trace class. It remains to see that invertible operators of the form $\one_L+x$ with $x\in\tcl$ form a group. First note that $(\one_L+x)(\one_L+y)=\one_L+x+y+xy$ and if $x,y\in\tcl$ then $x+y+xy\in\tcl$. Now, the inverse is a holomorphic function away from the origin. So $\spec((\one_L+x)^{-1})=\spec(\one_L+x)^{-1}$. Finally, the possible accumulation point $1$ is mapped to itself, approximately isometrically. So $(\one_L+x)^{-1}-\one_L$ is trace class.
\end{proof}

We define $H$ to be the set of operators on $\cF$ arising from the exponentiation of $\hat{\lambda}\in\fh$ where $\lambda\in\tcl$ and $\lambda^\la=-\lambda$. Note that this is the same as restricting to the unitary elements of $H^\C$, as a suitable branch of the logarithm can always be chosen.

To define the groups $\Ge$ and $G$ one strategy is to use the representations $\fn^\C$ and $\Wo$ to this end. Exponentiating the Lie algebra representation condition (\ref{eq:lainvact}) for $\fge^\C$ and $\fg^\C$ yields,
\begin{equation}
 Z^\wa=\sigma Z^{-1} \sigma,
\label{eq:gcinvact}
\end{equation}
as a possible definition of $\Ge^\C$ and $G^\C$ in terms of operators $Z$ on $\fn^\C$ and $\Wo$ respectively. The real groups $\Ge$ and $G$ arise then from the additional restriction of unitarity, i.e., $Z^\wa=Z^{-1}$. A weakness of this approach is that it does not immediately yield a description of the action of the groups on the Fock space $\cF$. After all, $\Ge$ and $G$ should really be groups of operators on $\cF$. Even if $L$ is finite-dimensional, it is necessary to exclude a singular subset from $\Ge$ to describe its action on $\cF$ \cite{Ber:grossneveu}. (In Berezin's article this arises as an invertibility condition on a submatrix of the representation matrix.) If $L$ is infinite-dimensional there are additional difficulties. $G$ and $\Ge$ defined in this way would almost certainly be too big, as e.g.\ the trace class restrictions on the operators defining $\fm$ and $\fn$ have no counterpart on the group side.

We shall use a more algebraic definition of $G$, directly in terms of the action on $\cF$. Define $\Galg^\C$ to be the group consisting of all finite products of exponentials of elements of $\fg^\C$ as operators on $\cF$. Let $\Galg'^\C$ be the completion of $\Galg^\C$ in the operator norm. Define $G^\C$ as the intersection of $\Galg'^\C$ with the invertible operators on $\cF$. Note that $G^\C$ is indeed a group and the subgroup $\Galg^\C\subseteq G^\C$ is dense. Moreover, the elements of $G^\C$ satisfy the condition (\ref{eq:gcinvact}) on $\Wo$.

It will be convenient to establish identities between certain products of exponentials of operators in $\fg^\C$. These are obtained from the relations of Section~\ref{sec:CAR} using suitable versions of Baker-Campbell-Hausdorff formulas.
\begin{lem}
Let $\lambda\in\tcl$. Let $\Lambda\in\asl$. Then, $\Lambda':L\to L$ defined by,
\begin{equation}
 \Lambda'\defeq e^{-\lambda} \Lambda e^{\lambda^\la}
\end{equation}
is in $\asl$. Moreover, we have the identity of operators on $\cF$,
\begin{equation}
 e^{\hat{\lambda}} e^{\hat{\Lambda}}=e^{\hat{\Lambda'}} e^{\hat{\lambda}} .
\label{eq:rellL}
\end{equation}
\end{lem}

\begin{lem}
Let $\lambda\in\tcl$ and $\xi\in L$. Define $\xi'\defeq e^{-\lambda} \xi$. Then,
\begin{equation}
e^{\hat{\lambda}} e^{\hat{\xi}} = e^{\hat{\xi'}} e^{\hat{\lambda}} .
\label{eq:rellx}
\end{equation}
\end{lem}

\begin{lem}
Let $\xi_1,\xi_2\in L$ and $\Lambda_1,\Lambda_2\in\asl$. Define
\begin{equation}
\Lambda'(\eta)\defeq \xi_2\{\eta,\xi_1\}-\xi_1\{\eta,\xi_2\}.
\end{equation}
Then $\Lambda'\in\asl$ and,
\begin{equation}
 e^{\hat{\Lambda_1}+\hat{\xi_1}} e^{\hat{\Lambda_2}+\hat{\xi_2}} = e^{\hat{\Lambda_1}+\hat{\Lambda_2}+\hat{\Lambda'}/2+\hat{\xi_1}+\hat{\xi_2}} .
\label{eq:relpm}
\end{equation}
\end{lem}

\begin{lem}
\label{lem:relpp}
Let $\xi,\xi'\in L$ and $\Lambda,\Lambda'\in\asl$. Assume that $\one-\Lambda'\Lambda$ is invertible. Then there is $\mu\in\tcl$ such that
\begin{equation}
 e^{-\mu}=\one - \Lambda'\Lambda .
\label{eq:ppmu}
\end{equation}
Set $b\defeq\{\xi,e^\mu\xi'\}$ and assume furthermore that $b\neq -2$. Then, there is $\beta\in\tcl$ such that
\begin{equation}
 e^{-\frac{1}{2}\beta}=\one+\frac{1}{2}\xi'\{e^{\mu^\la}\xi,\cdot\} .
\label{eq:ppbeta}
\end{equation}
Define moreover,
\begin{align}
 \nu & \defeq \frac{1}{2} \Lambda e^\mu \xi'\{\xi',\cdot\} \label{eq:ppnu} \\
 \nu'& \defeq \frac{1}{2} e^\mu \Lambda'\xi \{\xi,\cdot\} \label{eq:ppnup} \\
 \eta & \defeq 2\left(2+\overline{b}\right)^{-1}\left(e^{\mu^*}\xi-\Lambda e^\mu \xi'\right) \\
 \eta' & \defeq 2\left(2+b\right)^{-1}\left(e^{\mu}\xi'-e^\mu\Lambda' \xi\right) \\
 \Omega & \defeq \Lambda e^\mu -\left(2+\overline{b}\right)^{-1}\left(
 \Lambda e^\mu \xi'\{\cdot,e^{\mu^\la}\xi\}-e^{\mu^\la}\xi\{\cdot,\Lambda e^\mu\xi'\}\right) \\
 \Omega' & \defeq  e^\mu \Lambda' -\left(2+b\right)^{-1}\left(
 e^\mu \Lambda' \xi\{\cdot,e^{\mu}\xi'\}-e^{\mu}\xi'\{\cdot,e^\mu \Lambda' \xi\}\right)
\end{align}
Then, $\nu,\nu'\in\tcl$, $\Omega,\Omega'\in\asl$, and
\begin{equation}
 e^{\hat{\Lambda'}+\hat{\xi'}}e^{\hat{\Lambda}^\fa+\hat{\xi}^\fa}= e^{\hat{\Omega}^\fa+\hat{\eta}^\fa} e^{\hat{\nu}^\fa} e^{\hat{\beta}} e^{\hat{\mu}} e^{\hat{\nu'}} e^{\hat{\Omega'}+\hat{\eta'}} .
\label{eq:relpp}
\end{equation}
\end{lem}

Given the preceding operator identities allows us crucially to describe the structure of $G^\C$ in terms of a decomposition. Consider to this end the following definitions. Set $\fp_+\defeq\fm_+\oplus\fn_+$ and $\fp_-\defeq\fm_-\oplus\fn_-$.
Define $P_-, P_+$ to be the image of the exponential map applied to $\fp_-$ and $\fp_+$ respectively.
It follows immediately that $P_-,P_+$ are groups since $\fp_-$ and $\fp_+$ have the simple composition identities given by equation (\ref{eq:relpm}).

\begin{thm}[Decomposition Theorem]
The map
\begin{equation}
 P_-\times H^\C \times P_+ \to G^\C
\label{eq:gcdeca}
\end{equation}
given by the product in $G^\C$ is injective and its image dense in the operator norm topology.
\end{thm}
\begin{proof}
Injectivity is clear. For the denseness reformulate the statement as follows: Given $g\in G^\C$ and $\epsilon>0$ there are $\lambda\in\tcl$, $\Lambda,\Lambda'\in\asl$ and $\xi,\xi'\in L$ such that $\lop g-g'\rop<\epsilon$, where
\begin{equation}
 g'\defeq e^{\hat{\Lambda}^\fa+\hat{\xi}^\fa} e^{\hat{\lambda}} e^{\hat{\Lambda'}+\hat{\xi'}} .
\label{eq:gcdec}
\end{equation}
We note first that any element of $G^\C$ is arbitrarily approximated by products of exponentials with the special property that the argument of each exponential is in one of the subspaces $\fh^\C$, $\fp_+$, $\fp_-$ of $\fg^\C$. This follows by repeated application of Trotter's formula \cite{Suz:gentrotter}. We can thus use the identities (\ref{eq:rellL}), (\ref{eq:rellx}), (\ref{eq:relpp}) and their adjoint versions to ``move'' all exponentials of $\fp_-$ to the left, and all exponentials of $\fp_+$ to the right, in finitely many steps. The exponentials of $\fp_-$ on the left are then combined by the adjoint of equation (\ref{eq:relpm}) and those of $\fp_+$ on the right by equation (\ref{eq:relpm}). Finally, the exponentials of $\fh$ in the middle are combined by the group property of $H^\C$. The result is a decomposition of the required form (\ref{eq:gcdec}). There is an obstruction, however. Namely, the identity (\ref{eq:relpp}) fails for very special values of the arguments of the exponentials. In this case we slightly tweak the arguments of the exponentials on the left hand side of the identity so that the right hand side is well defined. Due to continuity this only slightly modifies the group element in terms of the operator norm topology. Since we have to do this only a finite number of times, we can still arbitrarily approximate the original group element. This completes the proof.
\end{proof}

We denote the image of the map (\ref{eq:gcdeca}) in the decomposition theorem by $\Gdec^\C\subseteq\Galg^\C$.
Instead of defining the real subgroup $G$ similarly, it will be more convenient to derive its definition from that of $G^\C$. Thus, define $G$ as the intersection of $G^\C$ with the unitary operators on $\cF$.
Note that $G$ is a group and that it is closed in the operator norm topology. Moreover, if $L$ is finite-dimensional, of dimension $n$, $G$ can be identified with the Lie group $\SO(2n+1)$. We also set $\Gdec\defeq \Gdec^\C\cap G$. This intersection can be characterized as follows.

\begin{thm}
\label{thm:gdec}
Let $r,s\in\tcl$ such that $r^*=r$ and $s^*=-s$. Let $\Lambda,\Lambda'\in\asl$ and $\xi,\xi'\in L$. Set
\begin{equation}
 e^{-\mu}=1-\Lambda^2 .
 \label{eq:gdmu}
\end{equation}
This determines $-\mu\in\tcl$ uniquely as a positive operator. Let $b\defeq \{\xi,e^\mu\xi\}$. Set
\begin{equation}
 e^{-\frac{1}{2}\beta}=\one+\frac{1}{2}\xi\{e^{\mu}\xi,\cdot\} .
 \label{eq:gdbeta}
\end{equation}
This determines $-\beta\in\tcl$ uniquely as a positive operator. Let $\nu\in\tcl$ be
\begin{equation}
 \nu\defeq \frac{1}{2} e^{\mu}\Lambda\xi \{\xi,\cdot\} .
 \label{eq:gdnu}
\end{equation}
 Then,
\begin{equation}
 g\defeq e^{\hat{\Lambda}^\fa+\hat{\xi}^\fa} e^{\hat{r}}e^{\hat{s}} e^{\hat{\Lambda'}+\hat{\xi'}}
\label{eq:gdec}
\end{equation}
is element of $G$ iff the following conditions are satisfied:
\begin{align}
 e^{-2\hat{r}} & = e^{\hat{\nu}^\fa} e^{\hat{\beta}} e^{\hat{\mu}} e^{\hat{\nu}} \label{eq:gd1}\\
 \Lambda' & = -e^s e^r\left(e^\mu \Lambda -\left(2+b\right)^{-1}\left(
 e^\mu \Lambda \xi\{\cdot,e^{\mu}\xi\}-e^{\mu}\xi\{\cdot,e^\mu \Lambda \xi\}\right)\right)e^{r} e^{-s}
 \label{eq:gd2}\\
 \xi' & = - 2\left(2+b\right)^{-1} e^s e^r (\one-\Lambda) e^\mu \xi . \label{eq:gd3}
\end{align}
\end{thm}
\begin{proof}(Sketch)
Imposing the unitarity condition leads to
\begin{equation}
e^{\hat{\Lambda'}^\fa+\hat{\xi'}^\fa} e^{\hat{s}^\fa}e^{\hat{r}^\fa} e^{\hat{\Lambda}+\hat{\xi}}
= e^{-\hat{\Lambda'}-\hat{\xi'}} e^{-\hat{s}}e^{-\hat{r}} e^{-\hat{\Lambda}^\fa-\hat{\xi}^\fa} .
\end{equation}
Rewrite this as,
\begin{equation}
e^{\hat{\Lambda'}^\fa+\hat{\xi'}^\fa} e^{\hat{s}^\fa}e^{\hat{r}^\fa} e^{\hat{\Lambda}+\hat{\xi}} e^{\hat{\Lambda'}+\hat{\xi'}}
= e^{-\hat{s}}e^{-\hat{r}} e^{-\hat{\Lambda}^\fa-\hat{\xi}^\fa} .
\end{equation}
Apply identity (\ref{eq:relpp}) to the two rightmost terms on the left hand side. Then, reorder and combine on each side separately via identities (\ref{eq:rellL}), (\ref{eq:rellx}), (\ref{eq:relpm}) and their adjoint versions. Finally, use the uniqueness of the decomposition to equate suitable terms on both sides.
\end{proof}

\begin{cor}
\label{cor:gdecparam}
There is a surjective map $\alpha:\fp_-\oplus\fh\to \Gdec$ such that for $\Lambda\in\asl$, $\xi\in L$, $s\in\tcl$ with $s^*=-s$ the image of $(\hat{\Lambda}^\fa+\hat{\xi}^\fa,\hat{s})$ is expression (\ref{eq:gdec}) where $r,\Lambda',\xi'$ are determined by conditions (\ref{eq:gd1}), (\ref{eq:gd2}), (\ref{eq:gd3}).
\end{cor}

Essentially all considerations (including the definitions) restrict straightforwardly to the subgroup $\Ge^\C\subset G^\C$. In particular, the Decomposition Theorem reduces to a simpler Decomposition Theorem with $P_+,P_-$ restricted to $M_+,M_-$, the exponentials of $\fm_+,\fm_-$. This restricted Decomposition Theorem is an infinite-dimensional generalization of a Decomposition Theorem of Balian and Brezin \cite{BaBr:nbogtrans}.

\section{Homogeneous spaces for coherent states}
\label{sec:homspace}

The coherent states in the approaches of Gilmore \cite{ZFG:coherent} and Perelomov \cite{Per:coherent} are the states that arise from the unitary action of the real dynamical group on the vacuum state. These are (up to phase factors) in correspondence to the elements of the homogeneous space obtained by dividing out the subgroup that maps the vacuum to a multiple of itself. We consider here the full dynamical group $G$. The relevant subgroup is then precisely $H$ which acts for $\lambda\in\tcl$ as
\begin{equation}
 \exp(\hat{\lambda})\vac=\exp\left(-\frac{1}{2}\tr_L(\lambda)\right)\vac .
\end{equation}
Thus the space of coherent states can be identified, up to phase factors, with the homogeneous space $G/H$. If we restrict to the even coherent states the relevant homogeneous space is $\Ge/H$, a submanifold of $G/H$.

We proceed to clarify the general structure and some properties of the homogeneous space $G/H$. To this end we shall use the language of differential geometry which in part is only appropriate for finite dimensional manifolds. We return to a language that is manifestly well defined also in the infinite-dimensional case at the end of this section. We start by considering the decomposition of the Lie algebras associated to the inclusions $H\subset \Ge$ and $H\subset G$. For the even dynamical Lie algebra $\fge=\fh\oplus\fm$ we have from the relations in Section~\ref{sec:CAR},
\begin{equation}
 [\fh,\fh]\subseteq\fh,\quad [\fm,\fm]\subseteq\fh,\quad [\fm,\fh]\subseteq \fm .
\label{eq:decfge}
\end{equation}
This implies that the map $\fge\to\fge$ which is the identity on $\fh$ and minus the identity on $\fm$ is an automorphism of $\fge$. This automorphism lifts to $\Ge$ and makes $\Ge/H$ into a \emph{symmetric space} \cite{Hel:dglgss}. This gives $\Ge/H$ a rich structure and ensures among other things that it is a Kähler manifold. In the context of coherent state this homogeneous space was analyzed in some detail by Berezin, including remarks on the infinite-dimensional case \cite{Ber:grossneveu}.

The full dynamical Lie algebra $\fg$ decomposes as $\fg=\fh\oplus\fm\oplus\fn$. Additional to (\ref{eq:decfge}) we have,
\begin{equation}
 [\fn,\fn]\subseteq\fh\oplus\fm,\quad [\fm,\fn]\subseteq\fn,\quad [\fh,\fn]\subseteq \fn .
\end{equation}
It will be convenient to denote $\fp\defeq \fm\oplus\fn$ so $\fg=\fh\oplus\fp$. Observe that we do not have $[\fp,\fp]\subseteq \fh$. So we cannot introduce an automorphism as above and $G/H$ is not a symmetric space. However, we still have $[\fh,\fp]\subseteq \fp$, implying that $G/H$ is a \emph{reductive homogeneous space} \cite{KoNo:diffgeo2}.

We recall how $G$, being a compact Lie group, acquires the structure of a Riemannian manifold with biinvariant metric.
Denote the left action of $G$ on itself by $L^g: G\to G, a\mapsto g a$  and the right action by $R^g: G\to G, a\mapsto a g$. We canonically identify the Lie algebra $\fg$ with the tangent space $T_e G$ at the unit element $e$ of $G$.
Denote by $x_a\in T_a G$ the left-invariant vector field generated by $x\in\fg=T_e G$ at $a\in G$, i.e., $x_a\defeq\xd L^a_e(x)$. Taking the inner product (\ref{eq:gip}) yields a Riemannian metric $d$ on $G$ via,
\begin{equation}
 d_a(x_a,y_a)\defeq \lhs x,y \rhs .
\end{equation}
This metric is left-invariant by construction.
The inner product (\ref{eq:gip}), being a multiple of the Killing form, is invariant under the adjoint action of $\fg$. Extending this to the Lie group $G$, it is invariant under the adjoint action of $G$, i.e., for $g\in G$, $x,y \in \fg$,
\begin{equation}
\lhs \ad_{g}(x), \ad_{g}(y)\rhs = \lhs x, y\rhs .
\end{equation}
The right action by $G$ on left-invariant vector fields is via the adjoint action on the Lie algebra, i.e,
\begin{equation}
\xd R^g_a(x_a)= (\ad_{g^{-1}}(x))_{ag} .
\end{equation}
Thus right action by $g\in G$ leaves the metric invariant as well,
\begin{equation}
 d_{ag}(\xd R^g_a(x_a),\xd R^g_a(y_a))=d_{ag}(x_{ag},y_{ag}) .
\end{equation}

The homogeneous space $G/H$, being reductive, acquires a left-invariant Riemannian metric from $G$ as follows.
Fix a point $p_0$ in $G/H$ and consider the associated projection
\begin{equation}
 \pi: G\to G/H,\quad g\mapsto g\, p_0 .
\end{equation}
The homogeneous space $G/H$ is then characterized by the relation $\pi(g)=\pi(g')$ iff there is $h\in H$ with $g'=g h$. We can write this relation as $\pi\circ R^h=\pi$ for all $h\in H$. For the tangent spaces this implies, $\xd\pi_g=\xd\pi_{g h}\circ \xd R^h_g$ as a map $T_gG\to T_{\pi(g)} G/H=T_{\pi(gh)} G/H$. Note that with respect to the invariant inner product (\ref{eq:gip}) in $\fg$ the subspaces $\fh$ and $\fp$ are orthogonal. Denote the orthogonal projection to $\fp$ by $r:\fg\to\fp$.
For $g\in G$ define a positive definite inner product $d'_{\pi(g)}$ on $T_{\pi(g)}G/H$ as follows,
\begin{equation}
 d'_{\pi(g)}(\xd\pi_g(x_g),\xd\pi_g(y_g))\defeq \lhs r(x),r(y)\rhs .
\end{equation}
This is well defined because $\xd\pi_g$ is surjective on $T_{\pi(g)}G/H$ and if $\xd\pi_g(x_g)=\xd\pi_g(x'_g)$ then $x-x'\in \fh$ and so $r(x-x')=0$. As the notation suggests, the inner product $d'$ does not actually depend on the chosen point $g\in G$. To see this, consider $h\in H$ and note,
\begin{multline}
 d'_{\pi(gh)}(\xd\pi_g(x_g),\xd\pi_g(y_g))
 =d'_{\pi(gh)}(\xd\pi_{gh}\circ \xd R^h_g (x_g),\xd\pi_{gh}\circ \xd R^h_g (y_g)) \\
 =\lhs(r(\ad_{h^{-1}}(x)),r(\ad_{h^{-1}}(y))\rhs= \lhs r(x),r(y) \rhs .
\end{multline}
Here in addition to the invariance of the inner product on $\fg$ under the adjoint action we have crucially used that the adjoint action commutes with $r$. This is true precisely due to the condition $[\fh,\fp]\subseteq \fp$ (in addition to $[\fh,\fh]\subseteq \fh$). The so defined Riemannian metric on $G/H$ is left-invariant under the action of $G$, again by construction.

Define the complex linear map $I:\fp^\C\to\fp^\C$ by $I x=\pm\im x$ for $x\in \fp_{\pm}$. Then, $I^2=-\one_{\fp^\C}$, i.e. $I$ is a complex structure on $\fp^\C$, different from the standard one given by multiplication by $i$. Observe from equation (\ref{eq:gipeval}) that $\fp_+$ and $\fp_-$ are orthogonal under the inner product. Thus, $I$ leaves the inner product invariant,
\begin{equation}
 \lhs I x, I y\rhs =\lhs x, y\rhs,\quad\forall x,y\in\fp^\C .
\end{equation}
What is more, by construction $I$ commutes with complex conjugation in $\fp^\C$, i.e., with the operation of taking minus the adjoint. Thus, $I$ restricts to a complex structure on the real vector space $\fp$. Indeed, it coincides with the one inherited from the natural complex structures on $L$ and $\asl$ via the maps $\xi\mapsto \hat{\xi}-\hat{\xi}^\fa$ and $\Lambda\mapsto \hat{\Lambda}-\hat{\Lambda}^\fa$.
$I$ gives thus rise to a Kähler polarization \cite{KoNo:diffgeo2}. The associated symplectic form is,
\begin{equation}
 \omega (x,y)\defeq\frac{1}{2}\lhs I x, y\rhs
\end{equation}
Parametrizing in the usual way, we get,
\begin{multline}
 \omega( \check{\Lambda_1}-\check{\Lambda_1}^\wa+\check{\xi_1}-\check{\xi_1}^\wa,
  \check{\Lambda_2}-\check{\Lambda_2}^\wa+\check{\xi_2}-\check{\xi_2}^\wa) \\
 =\frac{\im}{2}\left(\tr_L(\Lambda_2\Lambda_1)-\tr_L(\Lambda_1\Lambda_2)\right) +\im\left(\{\xi_2,\xi_1\} -\{\xi_1,\xi_2\}\right).
\end{multline}
The real inner product and the symplectic form combine to a complex inner product on $\fp$,
\begin{equation}
 s(x,y)\defeq \lhs x,y\rhs+2\im \omega(x,y) .
\end{equation}
Concretely, this is,
\begin{multline}
 s( \check{\Lambda_1}-\check{\Lambda_1}^\wa+\check{\xi_1}-\check{\xi_1}^\wa,
  \check{\Lambda_2}-\check{\Lambda_2}^\wa+\check{\xi_2}-\check{\xi_2}^\wa) \\
 =-2\tr_L(\Lambda_2\Lambda_1) +4\{\xi_1,\xi_2\}.
\end{multline}

We also note from relations (\ref{eq:mrel1}) and (\ref{eq:mrel2}) that $[\fh^\C,\fm_\pm]\subseteq \fm_\pm$ and from relations (\ref{eq:nrel1}) and (\ref{eq:nrel2}) that $[\fh^\C,\fn_\pm]\subseteq \fn_\pm$. Therefore $I$ is invariant under the adjoint action of $\fh^\C$, i.e., $[h,I x]=I [h,x]$ for $x\in\fp^\C$ and $h\in \fh^\C$. This implies that we can lift $I$ to the tangent bundle of $G/H$. For $g\in G$ define the linear map $I_{\pi(g)}:T_{\pi(g)}G/H\to T_{\pi(g)} G/H$ via,
\begin{equation}
 I_{\pi(g)}(\xd\pi_g(x_g))\defeq \xd\pi_g((I(r(x)))_g) .
\end{equation}
The same arguments as for the metric $d'$ show that $I_{\pi(g)}$ is well-defined, independent of the choice of $g$ and left-invariant. We can thus lift the structures $\omega$ and $s$ as well and obtain that $G/H$ becomes a Kähler manifold with $G$-invariant Kähler structure.
Under the action of the subgroup $\Ge\subset G$, $G/H$ decomposes into a union of orbits which are copies of the homogeneous space $\Ge/H$. Since the complex structure $I$ on $\fp$ leaves the subspace $\fm\subset\fp$ invariant, $\Ge/H$ inherits in this way a Kähler structure from $G/H$, coinciding with its known Kähler structure.

We shall now take advantage of the structural description of $G^\C$ and $G$ that we have obtained in Section~\ref{sec:gparam} to provide a satisfactory definition of the homogeneous space $G/H$ also in the infinite-dimensional case. To this end we note that the decomposition (\ref{eq:gcdeca}) is precisely adapted to making the action of the dynamical group on the vacuum state tractable. Namely, given $r,s\in\tcl$ with $r$ self-adjoint, $s$ skew-adjoint, $\Lambda,\Lambda'\in\asl$, $\xi,\xi'\in L$ we have,
\begin{equation}
e^{\hat{\Lambda}^\fa+\hat{\xi}^\fa} e^{\hat{r}} e^{\hat{s}} e^{\hat{\Lambda'}+\hat{\xi'}} \vac
= e^{-\frac{1}{2}\tr(r)} e^{-\frac{1}{2}\tr(s)} e^{\hat{\Lambda}^\fa+\hat{\xi}^\fa} \vac .
\label{eq:gcactvac}
\end{equation}
It is therefore convenient to take $\Gdec$ rather than $G$ as the starting point to define the homogeneous space. The usefulness of the parametrization of $\Gdec$ in terms of $\fp_-\oplus\fh$ of Corollary~\ref{cor:gdecparam} is also obvious. What is more, the decomposition also allows us to realize the homogeneous space as an actual subspace of $\Gdec$, i.e., it allows us to systematically choose one representative in each equivalence class. The equality, for $\Gdec$,
\begin{equation}
e^{\hat{\Lambda}^\fa+\hat{\xi}^\fa} e^{\hat{r}} e^{\hat{s}} e^{\hat{\Lambda'}+\hat{\xi'}}
= e^{\hat{\Lambda}^\fa+\hat{\xi}^\fa} e^{\hat{r}} e^{\hat{\Lambda''}+\hat{\xi''}} e^{\hat{s}}
\end{equation}
where $\Lambda''$ and $\xi''$ are functions of $\Lambda'$, $\xi'$ and $s$ shows that such a choice is given by setting $s=0$. We shall from now on use the notation $G/H$ to denote this specific realization of the homogeneous space. This definition also resolves the phase ambiguity we would in general have in identifying the space of coherent states with $G/H$.
\begin{cor}
\label{cor:qparam}
There is a bijective map $\fp_-\to G/H$ such that for $\Lambda\in\asl$, $\xi\in L$, the image of $\hat{\Lambda}^\fa+\hat{\xi}^\fa$ is expression (\ref{eq:gdec}) where $r,\Lambda',\xi'$ are determined by conditions (\ref{eq:gd1}), (\ref{eq:gd2}), (\ref{eq:gd3}) and $s=0$.
\end{cor}

The decomposition (\ref{eq:gcdeca}) also suggests a slightly different perspective on coherent states. Namely, we may consider the action of the complexified dynamical group $\Gdec^\C$ on the vacuum vector rather than that of its real version. Recalling equation (\ref{eq:gcactvac}) we see that the subgroup leaving invariant the vacuum state is now the group $X_+$ generated by $H^\C$ and $P_+$. Thus the corresponding homogeneous space parametrizing coherent states is $G^\C/X_+$. Using the decomposition (and restricting to $\Gdec^\C$) suggests to identify this with the subgroup $P_-\subset G^\C$ which amounts to a choice of representative for each element of $G^\C/X_+$. We can then parametrize $G^\C/X_+$ through $\fp_-$ via the exponential map. While the homogeneous spaces $G/H$ and $G^\C/X_+$ constructed this way are both parametrized by $\fp_-$ and are isomorphic, the realizations as subgroups that we have chosen are not identical. In particular, $P_-$ is not a subgroup of $G$. The actions on the vacuum of $G/H$ and $G^\C/X_+$ realized in this way differ only by a multiplicative factor, however. Namely, the term $\exp(\hat{r})$ is present only for $G/H$. It enforces unitarity by normalizing the resulting vector to have unit norm.

\section{Coherent states in the Fock space}
\label{sec:cohspace}

Recall that $\fp_-$ carries a positive definite inner product by restriction of expression (\ref{eq:gipeval}),
\begin{equation}
 \lhs (\Lambda_1,\xi_1), (\Lambda_2,\xi_2) \rhs =-\tr_L(\Lambda_1\Lambda_2) +2\{\xi_2,\xi_1\}.
\end{equation}
Here as in the following we write $(\Lambda,\xi)$ for the element $\hat{\Lambda}^\fa+\hat{\xi}^\fa$ of $\fp_-$. Remember also that the complex structure on $\fp_-$ is opposite to that of $\asl\oplus L$ under this natural identification. The inner product makes $\fp_-$ into a Hilbert space.
Define the map $K:\fp_-\to \cF$ as the exponential composed with application to the vacuum vector,
\begin{equation}
 (\Lambda,\xi) \mapsto \exp\left(\hat{\Lambda}^\fa+\hat{\xi}^\fa\right)\vac .
\end{equation}
As discussed in the previous section this gives rise to coherent states in terms of the homogeneous space $G^\C/X_+$ as identified with the subgroup $P_+\subset G^\C$.

\begin{prop}
The map $K$ is continuous, holomorphic and injective.
\end{prop}
\begin{proof}
First observe that the norm coming from the inner product $\lhs\cdot,\cdot\rhs$ on $\fp_-$ is not the same as the operator norm, but is equivalent to it. Comparing equation (\ref{eq:gipeval}) to equations (\ref{eq:Lopn}) and (\ref{eq:xopn}) yields 
\begin{equation}
2 \lop x\rop^2 \le \|x\|_{\fp_-}^2 \le 4 \lop x\rop^2 .
\end{equation}
Since the exponential map as well as evaluation on a state are continuous in the operator norm, $K$ is continuous. Also, the exponential map as well as evaluation on a state are holomorphic. So $K$ is holomorphic. Now recall from equations (\ref{eq:Lopn}) and (\ref{eq:xopn}) that for $(\Lambda,\xi)$ to act trivially on the vacuum we must have $\Lambda=0$ and $\xi=0$. But this faithful action on the vacuum can be recovered from the exponential by restricting to the subspaces of the Fock space of degree $1$ and $2$. Thus, $K$ is injective.
\end{proof}

Recall the map $\alpha:\fp_-\oplus\fh\to\Gdec$ from Corollary~\ref{cor:gdecparam}. Define $\tilde{K}:\fp_-\to\cF$ by,
\begin{equation}
 (\Lambda,\xi) \mapsto \alpha\left(\hat{\Lambda}^\fa+\hat{\xi}^\fa\right)\vac .
\end{equation}
From equation (\ref{eq:gcactvac}) we see that this action takes the form
\begin{equation}
 \alpha\left(\hat{\Lambda}^\fa+\hat{\xi}^\fa\right)\vac =  \exp\left(-\frac{1}{2}\tr(r)\right)\exp\left(\hat{\Lambda}^\fa+\hat{\xi}^\fa\right)\vac,
\end{equation}
where $r\in\tcl$ is self-adjoint and determined by equation (\ref{eq:gd1}). In particular $\tr(r)$ is real and $\exp\left(-\frac{1}{2}\tr(r)\right)$ is strictly positive. Since $\Gdec$ acts unitarily, comparing $K$ and $\tilde{K}$ yields,
\begin{equation}
 \tilde{K}(\Lambda,\xi)=\frac{K(\Lambda,\xi)}{\|K(\Lambda,\xi)\|_{\cF}},\qquad \|K(\Lambda,\xi)\|_{\cF} = \exp\left(\frac{1}{2}\tr_L(r)\right) .
\end{equation}
We can evaluate the norm of $K(\Lambda,\xi)$ in terms of a Fredholm determinant.

\begin{prop}
\label{prop:Knorm}
Let $\Lambda\in\asl$, $\xi\in L$. Then,
\begin{equation}
\|K(\Lambda,\xi)\|_{\cF}=\left(1+\frac{1}{2}b\right)^{\frac{1}{2}}\det \left(\one_L-\Lambda^2\right)^{\frac{1}{4}} ,\qquad b\defeq \{\xi,(\one_L-\Lambda^2)^{-1}\xi\} .
\end{equation}
Moreover,
\begin{equation}
1\le \|K(\Lambda,\xi)\|_{\cF} \le \exp\left(\frac{1}{4}\|(\Lambda,\xi)\|_{\fp_-}\right) .
\end{equation}
\end{prop}
\begin{proof}
We can reexpress the trace of a trace class operator in terms of the Fredholm determinant \cite{Sim:traceideals}. In the present case, this yields
\begin{equation}
\exp\left(\frac{1}{2}\tr_L(r)\right)=\det\left(e^{\frac{1}{2}r}\right)
=\det\left(e^{2r}\right)^{\frac{1}{4}} .
\end{equation}
From condition (\ref{eq:gd1}) we get $e^{2r} = e^{-\nu^\la} e^{-\beta} e^{-\mu} e^{-\nu}$ which in turn yields,
\begin{equation}
 \det\left(e^{2r}\right)^{\frac{1}{4}} =\det\left(e^{-\frac{1}{4}\nu^\la}\right)\det\left(e^{-\frac{1}{2}\beta}\right)^{\frac{1}{2}}\det\left(e^{-\mu}\right)^{\frac{1}{4}}\det\left(e^{-\frac{1}{4}\nu}\right)
\end{equation}
Now rewrite the $\nu$-terms again in terms of traces,
\begin{equation}
\det\left(e^{-\frac{1}{4}\nu}\right)=\exp\left(-\frac{1}{4}\tr_L(\nu)\right),
\end{equation}
an similarly for $\nu^\la$. From equation (\ref{eq:gdnu}) one easily finds that the trace of $\nu$ (and thus also its adjoint) vanishes. It remains to consider the terms containing $\beta$ and $\mu$. For the $\beta$-term recall equation (\ref{eq:gdbeta}), i.e., we have
\begin{equation}
 e^{-\frac{1}{2}\beta}=\one+\frac{1}{2}\xi\{e^{\mu}\xi,\cdot\} .
 \label{eq:calcbeta}
\end{equation}
For $n\ge 1$ we have,
\begin{equation}
 \left(\xi\{e^\mu\xi,\cdot\}\right)^n=b^{n-1}\xi\{e^\mu\xi,\cdot\} .
\end{equation}
Observe also $\tr_L(\xi\{e^{\mu}\xi,\cdot\})=b$.
Thus, we may set $-\frac{1}{2}\beta=t \xi\{e^\mu\xi,\cdot\}$. If $b=0$ then equation (\ref{eq:calcbeta}) is solved by $t=\frac{1}{2}$ but $\tr(\xi\{e^{\mu}\xi,\cdot\})$ vanishes and the determinant of expression (\ref{eq:calcbeta}) is unity. Suppose then that $b> 0$ ($b$ is necessarily non-negative since $e^\mu$ is a positive operator) and consider thus,
\begin{equation}
 \exp\left(t \xi\{e^\mu\xi,\cdot\}\right)=\sum_{n=0}^\infty \frac{1}{n!} t^n \left(\xi\{e^\mu\xi,\cdot\}\right)^n = \one_L + \frac{1}{b} \left(\exp(tb)-1\right) \xi\{e^\mu\xi,\cdot\} .
\end{equation}
This solves equation (\ref{eq:calcbeta}) for
\begin{equation}
\exp(tb)=1+\frac{1}{2}b .
\end{equation}
Then,
\begin{equation}
 \det\left(e^{-\frac{1}{2}\beta}\right)=\exp\left(\tr_L(t \xi\{e^\mu\xi,\cdot\})\right)^{\frac{1}{2}}
 =\exp(tb)^{\frac{1}{2}}=\left(1+\frac{1}{2}b\right)^{\frac{1}{2}}
\end{equation}
This also correctly captures the case $b=0$.
Finally, for the $\mu$-term we simply insert the defining relation (\ref{eq:gdmu}).

The lower estimate on $\|K(\Lambda,\xi)\|_{\cF}$ follows from $b\ge 0$ and the positivity of $-\Lambda^2$. For the upper estimate note,
\begin{multline}
\|K(\Lambda,\xi)\|_{\cF}^2 =\left(1+\frac{1}{2}b\right) \det\left(\one_L-\Lambda^2\right)^{\frac{1}{2}}
 \le \left(1+\frac{1}{2}\|\xi\|_L^2 \right) \exp\left(\frac{1}{2}\tr_L(-\Lambda^2)\right) \\
 \le \exp\left(\frac{1}{2}\|\xi\|_L^2-\frac{1}{2}\tr_L(\Lambda^2)\right)
 \le \exp\left(\|\xi\|_L^2-\frac{1}{2}\tr_L(\Lambda^2)\right) .
\end{multline}
Here we have used that $e^\mu\le \one_L$ since $e^{-\mu}\ge \one_L$ and an estimate for the Fredholm determinant in terms of the trace, $\|\det(\one+X)\|\le \exp(\tr(|X|))$ \cite{Sim:traceideals}.
\end{proof}

\begin{prop}
\label{prop:cohip}
Let $\Lambda_1,\Lambda_2\in\asl$ and $\xi_1,\xi_2\in L$. If $\one_L-\Lambda_1\Lambda_2$ is invertible set
\begin{equation}
b\defeq\{\xi_2,(\one_L-\Lambda_1\Lambda_2)^{-1}\xi_1\} .
\end{equation}
Then,
\begin{equation}
\langle K(\Lambda_1,\xi_1),K(\Lambda_2,\xi_2)\rangle=
 \left(1+\frac{1}{2}b\right)
\det \left(\one_L-\Lambda_1\Lambda_2\right)^{\frac{1}{2}}.
\label{eq:kip}
\end{equation}
The correct branch of the square root is obtained by analytic continuation from $\Lambda_1=\Lambda_2$. If $\one_L-\Lambda_1\Lambda_2$ is not invertible, the inner product vanishes.
\end{prop}
\begin{proof}
We use Lemma~\ref{lem:relpp} to evaluate the inner product.
\begin{align}
\langle K(\Lambda_1,\xi_1),K(\Lambda_2,\xi_2)\rangle & =
\langle \vac, e^{\hat{\Lambda_1}+\hat{\xi_1}}e^{\hat{\Lambda_2}^\fa+\hat{\xi_2}^\fa}\vac\rangle \\
& = \langle \vac, e^{\hat{\Omega}^\fa+\hat{\eta}^\fa} e^{\hat{\nu}^\fa} e^{\hat{\beta}} e^{\hat{\mu}} e^{\hat{\nu'}} e^{\hat{\Omega'}+\hat{\eta'}} \vac\rangle \\
& = \langle \vac, e^{\hat{\nu}^\fa} e^{\hat{\beta}} e^{\hat{\mu}} e^{\hat{\nu'}} \vac\rangle \\
& = \exp\left(-\frac{1}{2}\tr_L(\nu^\la+\beta+\mu+\nu')\right)
\end{align}
Here $\nu,\nu',\beta,\mu$ are determined by equations (\ref{eq:ppnu}), (\ref{eq:ppnup}), (\ref{eq:ppbeta}), (\ref{eq:ppmu}) with $(\Lambda,\xi)=(\Lambda_2,\xi_2)$ and $(\Lambda',\xi')=(\Lambda_1,\xi_1)$. The procedure to evaluate the exponential of the traces is rather analogous to the one followed in the proof of Proposition~\ref{prop:Knorm}. We omit the details. Care has to be taken due to the fact that $\one_L-\Lambda_1\Lambda_2$ is not necessarily positive and might in fact be non-invertible. In the latter case the Fredholm determinant of this operator vanishes and so does the inner product. Also $b$ might take the exceptional value $-2$ in which case, again, the inner product vanishes. Furthermore, the correct branch has to be chosen for the square root in (\ref{eq:kip}).
\end{proof}

\begin{prop}
\label{prop:kdense}
The images of $K$ and $\tilde{K}$ span dense subspaces of $\cF$.
\end{prop}
\begin{proof}
Consider states of the Fock space $\cF$ of the form
\begin{equation}
 a_{\xi_1}^\fa a_{\xi_2}^\fa\cdots a_{\xi_n}^\fa\vac ,
\end{equation}
where $\xi_1,\dots,\xi_n$ are orthogonal and of unit norm in $L$. The space of linear combinations of such states is dense in $\cF$. We shall call these states \emph{generating states}. It is sufficient to show that any generating state is equal to a linear combination of coherent states. We consider states of even degree first and proceed by induction. The vacuum state $\vac$ is obviously equal to $K(0,0)$. Now suppose the statement is true for generating states of even degree up to $2n-2$. Consider now the generating state
\begin{equation}
 a_{\xi_1}^\fa a_{\xi_2}^\fa\cdots a_{\xi_{2n}}^\fa\vac ,
\end{equation}
Extend $\{\xi_1,\dots,\xi_{2n}\}$ to an orthonormal basis $\{\xi_i\}_{i\in I}$ of $L$. Define $\Lambda:L\to L$ as follows. Set for $k\in\{1,\dots,n\}$,
\begin{equation}
 \Lambda(\xi_{2k})\defeq\xi_{2k-1},\quad \Lambda(\xi_{2k-1})\defeq-\xi_{2k} .
\end{equation}
Moreover, set $\Lambda(\xi_l)\defeq 0$ if $l>2n$. Extend $\Lambda$ to all of $L$ conjugate linearly. It is easy to verify that $\Lambda$ is anti-symmetric and $\Lambda^2$ is trace class. What is more, in terms of degree $K(\Lambda,0)$ contains only contributions up to order $2n$ and the latter coincide with the generating state. Indeed,
\begin{equation}
 K(\Lambda,0)- a_{\xi_1}^\fa a_{\xi_2}^\fa\cdots a_{\xi_{2n}}^\fa\vac
\end{equation}
can be written as a linear combination of generating states of degree up to $2n-2$. This completes the induction.

We proceed to consider generating states of odd degree. Again we proceed by induction. We assume the statement to be true for generating states up to degree $2n$. Consider thus a generating state
\begin{equation}
 a_{\xi_1}^\fa a_{\xi_2}^\fa\cdots a_{\xi_{2n+1}}^\fa\vac .
\end{equation}
Again, extend $\{\xi_1,\dots,\xi_{2n+1}\}$ to an orthonormal basis of $L$ and define $\Lambda$ in exactly the same way as above. Then,
\begin{equation}
 K(\Lambda,\sqrt{2}\xi_{2n+1})- a_{\xi_1}^\fa a_{\xi_2}^\fa\cdots a_{\xi_{2n+1}}^\fa\vac
\end{equation}
is a linear combination of generating states of degree up to $2n$. This completes the proof.
\end{proof}

\begin{prop}
Given $\psi\in\cF$ define the function
\begin{equation}
f_\psi:\fp_-\to\C\quad\text{by}\quad f_\psi(\Lambda,\xi)\defeq \langle K(\Lambda,\xi),\psi\rangle.
\label{eq:defwf}
\end{equation}
Then, $f_\psi$ is continuous and anti-holomorphic.
\end{prop}
\begin{proof}
$f_\psi$ is continuous since $K$ is continuous and evaluation with the inner product is continuous. Moreover, $f_\psi$ is anti-holomorphic since $K$ is holomorphic and evaluation on the left with the inner product is conjugate linear.
\end{proof}

Let $\hol(\fp_-)$ denote the complex vector space of continuous and anti-holomorphic functions on $\fp_-$.

\begin{lem}
The complex linear map
\begin{equation}
 f:\cF\to\hol(\fp_-)\quad\text{given by}\quad
 \psi \mapsto f_\psi 
\end{equation}
is injective.
\end{lem}
\begin{proof}
The injectivity is ensured by the denseness of the coherent states, i.e., by Proposition~\ref{prop:kdense}.
\end{proof}

Let $F(\fp_-)\subseteq\hol(\fp_-)$ denote the image of $f$.

\begin{thm}
The complex linear isomorphism
\begin{equation}
 f:\cF\to F(\fp_-) 
\end{equation}
realizes the Fock space $\cF$ as a \emph{reproducing kernel Hilbert space} of continuous anti-holomorphic functions on the Hilbert space $\fp_-$ with reproducing kernel $K:\fp_-\times\fp_-\to\C$,
\begin{equation}
 K\left((\Lambda_1,\xi_1),(\Lambda_2,\xi_2)\right)= \langle K(\Lambda_1,\xi_1),K(\Lambda_2,\xi_2)\rangle
\end{equation}
given by equation (\ref{eq:kip}) of Proposition~\ref{prop:cohip}. In particular, the reproducing property is equation (\ref{eq:defwf}).
\end{thm}
\begin{proof}
This follows straightforwardly by combining the previous results.
\end{proof}

\subsection*{Acknowledgments}

This work was supported in part by UNAM--DGAPA--PASPA through a sabbatical grant, by UNAM--DGAPA--PAPIIT through project grant IN100212 and by the Emerging Fields project ``Quantum Geometry'' of the Friedrich-Alexander-Universität Erlangen-Nürnberg.

\bibliographystyle{stdnodoi} 
\bibliography{stdrefsb}
\end{document}